\documentclass[smallextended,natbib,runningheads]{svjour3}
\journalname{Annals of the Institute of Statistical Mathematics}
\smartqed  % flush right qed marks, e.g. at end of proof
\usepackage{graphicx,pproof}
\usepackage{graphicx}
\usepackage{amsmath,amsfonts,amssymb,color,url}
\usepackage{natbib}
\usepackage{booktabs}

\def\ve#1{\mathchoice{\mbox{\boldmath$\displaystyle\bf#1$}}
{\mbox{\boldmath$\textstyle\bf#1$}}
{\mbox{\boldmath$\scriptstyle\bf#1$}}
{\mbox{\boldmath$\scriptscriptstyle\bf#1$}}}

\usepackage{amsopn}
\usepackage{rotating}
\usepackage{graphpap}
\usepackage{array}
\usepackage{multirow}
\usepackage{fancyhdr}

\usepackage[margin=0.85in]{geometry}
\usepackage{graphics}
 \usepackage{graphicx}
\usepackage{hyperref}
\usepackage[normalem]{ulem}
\usepackage{slashbox} % to use backslashbox in table
\usepackage{subfigure} %to use subfigures
  \usepackage{amsfonts}
\usepackage{amsmath,amssymb}
\usepackage{makeidx}
\usepackage{url}
\graphicspath{{figs/}}
\usepackage{caption}

\usepackage{natbib}
\usepackage{color}

\title{Distributions of topological tree metrics between a species tree and a gene
tree}
\titlerunning{Distributions of topological tree metrics}

\author{Jing Xi \and  Jin Xie  \and Ruriko Yoshida}
\authorrunning{Xi, Xie, Yoshida}
\institute{J. Xi \at Department of Mathematics\\  
North Carolina State University\\ 
2108 SAS Hall, 2311 Stinson   Drive\\ 
Raleigh, NC 27695, USA\\
              \email{jxi2@ncsu.edu}           %  \\
           \and
                      J. Xie \at
Statistics Department\\
              University of Kentucky\\
Multidisplinary Science Building\\
Lexington, KY 40506-0082, USA\\
\email{jin.xie@uky.edu}
           \and
                      R. Yoshida \at
Correnponding Author\\
Statistics Department\\
              University of Kentucky\\
325D Multidisplinary Science Building\\
Lexington, KY 40506-0082, USA\\
\email{ruriko.yoshida@uky.edu}
}

\begin{document}

\def\ve#1{\mathchoice{\mbox{\boldmath$\displaystyle\bf#1$}}
{\mbox{\boldmath$\textstyle\bf#1$}}
{\mbox{\boldmath$\scriptstyle\bf#1$}}
{\mbox{\boldmath$\scriptscriptstyle\bf#1$}}}

%\newtheorem{theorem}{Theorem}%[section]
% \newtheorem{definition}[theorem]{Definition}
% \newtheorem{question}[theorem]{Question}
% \newtheorem{problem}[theorem]{Problem}
% \newtheorem{proposition}[theorem]{Proposition}
% \newtheorem{lemma}[theorem]{Lemma}
% \newtheorem{corollary}[theorem]{Corollary}
% \newtheorem{remark}[theorem]{Remark}
% \newtheorem{example}[theorem]{Example}
% \newtheorem{ex}[theorem]{Example}
% \newtheorem{algorithm}[theorem]{Algorithm}

%Rudy's newcommands
\newcommand{\Tn}{\mathcal T_n}
\newcommand{\Q}{{\mathbb Q}}
\newcommand{\Z}{{\mathbb Z}}
\newcommand{\N}{{\mathbb N}}
\newcommand{\X}{{\mathbb X}}
\newcommand{\C}{{\mathbb C}}
\newcommand{\R}{\mathbb R}
\newcommand{\E}{\mathbb E}
\newcommand{\V}{\mathbb V}
\newcommand{\ttt}{\mathbf{T}}
\pagestyle{headings}
\setcounter{page}{1}
%\pagenumbering{numeric}
\date{Received: date / Revised: date}
\maketitle                 % Produces the title.

\begin{abstract}
In order to conduct a statistical analysis on a given set of
phylogenetic gene trees, we often use a
distance measure between two trees.  In a statistical distance-based method to analyze
discordance between gene trees, it is a key to decide ``biologically 
meaningful'' and ``statistically well-distributed'' distance between
trees.  Thus, in this paper, we study the distributions of the three
tree distance metrics: the edge difference, the path difference, and the
precise $K$ interval cospeciation distance, between two trees: First,
we focus on distributions of the three tree distances between two
random unrooted trees with $n$ leaves ($n \geq 4$); and then we
focus on the distributions the three tree distances between a fixed
rooted species tree with $n$ leaves and a random gene tree with $n$
leaves generated under the coalescent process with the given species
tree. We show some theoretical results as well as simulation study on
these distributions.   
\end{abstract}

\section{Introduction}

A central issue in systematic biology is the reconstruction of
populations and species from numerous gene trees with varying levels
of discordance \citep{Brito2009,Edwards2009}. While there has been a
well-established understanding of the discordant phylogenetic
relationships that can exist among independent gene trees drawn from a
common species tree
\citep{Pamilo1988,Takahata1989,Maddison1997,Bollback2009},
phylogenetic studies have only recently begun to shift away from
single gene or concatenated gene estimates of phylogeny towards these
multi-locus approaches
(e.g. \citep{Carling2008,Yu2011b,Betancur2013,Heled2013,Thompson2013}).
In order to conduct a statistical analysis on the given set of gene
trees, we vectorize each tree, i.e., 
converting them into a numerical vector format based on a {\em distance
matrix} or {\em dissimilarity
map}.  These vectorized trees can then be analyzed
as points in a multi-dimensional space where the {\em distance between trees}
increases as they become more dissimilar \cite[]{Hillis:2005xc, Semple2003,
Graham:2010qa}.  Such statistical
applications that test for incongruence or congruence between two trees using a
measurement of dissimilarity between a pair of trees are
called {\em distance-based methods} (for example, 
\cite{Holmes:2007mb, Arnaoudova:2010uq,kdetrees} are such statistical
methods).   In a statistical distance-based method to analyze
discordance between gene trees, it is a key to decide ``biological
meaningful'' and ``statistically well-distributed'' distance between
trees \cite[]{Steel1993, Coons}.  Therefore we have studied the
distributions of some well-known tree distances between trees.  
In this paper we focus on three topological tree distances {\em edge
  difference distance} \cite[]{Clifford1971}, and {\em 
  precise $k$-Interval Cospeciation ($K$-IC) distance} \cite[]{cophy}, and the {\em
  path difference} \cite[]{Steel1993} while the distributions of {\em
  Robinson--Foulds} (RF)
distances \cite[]{Robinson1981} and {\em quartet distances}
\cite[]{Pedersen2001} between random trees are very
well studied (for example, \cite{Steel1993}).  

Here we have conducted simulation studies on these distributions 
and we have shown theoretical results on the distributions of these tree
distances between the {\em species
  tree} and {\em gene trees} which are generated under the {\em
  coalescent process} \cite[]{coal}. 

For the precise $K$-IC distance between two random trees, \cite{Coons}
showed that if we take the random trees and 
compute the distance between them and if we send the number of
leaves $n$ of the trees to infinity, then the probability that the
distance between two random trees becomes the worst possible distance,
that is  $(n - 3)$, goes to zero while the probability that the  
RF distance between two random trees becomes the worse possible, that
is $2n - 6$, goes to one (Theorem 8 in \cite{Coons}).  This proporty
is very important to have in terms of applying statistical analysis on
the distances of trees.  
In addition, \cite{Steel1993} showed some simulation study as well as some
theoretical study on the distributions of the RF distance, Quartet
distance and path difference 
distance between random trees with $n = 12$ leaves (see Figure 6 on
\cite{Steel1993}).  
A key ingredient of analyzing distributions of these three tree
distances between two random trees with $n$ leaves is a simple
observation that the precise $K$-IC distance between trees is
$l_{\infty}$ norm of two vectorized trees, the path difference
distance is $l_{2}$ norm of two vectorized trees, and the edge
difference distance is $l_{1}$ norm of two vectorized trees. 
First, in this paper, we will show some theoretical results comparing distributions of
these tree  distances between random trees with $n$ leaves.  

A coalescent process is often used to model gene trees given a fixed
species tree with $n$ 
leaves.   These 
theoretical developments have been used to reconstruct species trees
from samples of estimated gene trees in practice 
\cite[]{Maddison2006,Carstens2007,Edwards2007,Mossel2010,
  RoyChoudhury2008}.  
\cite{Rosenberg2002} studied the distribution
of the topological concordance of gene trees and species trees under
the coalescent process, \cite{Rosenberg2003} worked on the
distributions of monophyly, paraphyly, and polyphyly in a coalescent
model, and  \cite{Degnan2005} studied the distribution of gene trees
under the coalescent process.  
In this paper we focus on the distributions of the edge difference,
path difference, and precise $K$-IC distances between the fixed
species tree and gene trees generated under the coalescent process.   

This paper is organized as follows.  In Section \ref{notation} we
remind readers some definitions.  In
Section \ref{random_unrooted},
we focus on the distributions of these three tree distances between
two unrooted random trees.  More specifically, in Subsection
\ref{path}, we will show the variance of the distribution of 
the path difference distance between two random trees with $n$
leaves.  In Subsection \ref{edge} and \ref{kic_sec} we will compare
the means of the distributions of the edge difference and precise
$K$-IC distances between random trees with the mean of the
distribution on the path difference distance between them.  In Section
\ref{sp_gene}, we focus on the distributions of these three different
tree distances between a fixed species tree and a gene tree generated
from the coalescent process with the species tree.  Especially we
have computed explicitly the probability that the distribution of any of the three
tree distances between a fixed species tree and a gene tree generated
under the coalescent process. In Section \ref{simulations}, we have
shown several simulation studies on the distributions of the three
different tree distributions between random trees as well as between a
fixed species tree and a gene tree generated from the coalescent.  We
end with discussions in Section \ref{dis}.

\section{Basics and notation}\label{notation}

In the subsequent descriptions, let $n$ be the number of leaves 
(terminal taxa) in a tree. Let $\Tn$  be the space of all possible
unrooted trees
on $n$ taxa and let $\Tn'$ be the space of all possible rooted trees
on $n$ taxa. In this paper we consider only tree metrics between two
trees using topological information of the trees, i.e., this tree
space does not incorporate
branch length information.  We use $||\cdot||_p$ to represent the usual
$l_p$ norm of a vector, and $|\cdot|$ to indicate the cardinality
of a set. A tree distance is a function, $d: \Tn \times \Tn \to \R^+$
that has, at a minimum, the properties $d(r,s)=d(s,r)$ and
$d(t,t)=0$. Many of the methods also require a vectorization function,
$v : \Tn \to \R^m$, for some $m$, which maps phylogenetic trees into
Euclidean space. The symmetric difference between two sets is defined
as $A\ominus B := (A\backslash B) \cup (B\backslash A)$.

  Several popular tree distances are squared Euclidean distances as will
  be demonstrated below.  

The {\em dissimilarity map} or {\em distance matrix} of a tree $T$ is a $n\times n$
symmetric matrix of non-negative real
numbers, with zero diagonals and off diagonal elements corresponding to the sum of the branch lengths
between pairs of leaves in the tree.  % Since the distance matrix is
% symmetrix and their diagonals are zeros we only consider elements in
% the upper triangle, i.e., we can vectorize $T$
% by ordering lexicographically elements in the upper triangle of  the
% distance matrix such that:
% \[v(T) := (D_{12}(T), D_{13}(T), \ldots , D_{23}(T),\ldots,D_{n-1, n}(T)). \]

Suppose $v:\Tn \to \Z^{n\choose 2}$ is a function such that  the
$(i, j)$th coordinate, where $1 \leq i < j \leq n$, of the $v(T)$ is
the number of edges on the unique path  between leaves $i$ and $j$ on
$T$.  

\begin{figure}
\begin{center}
\includegraphics[width=4in]{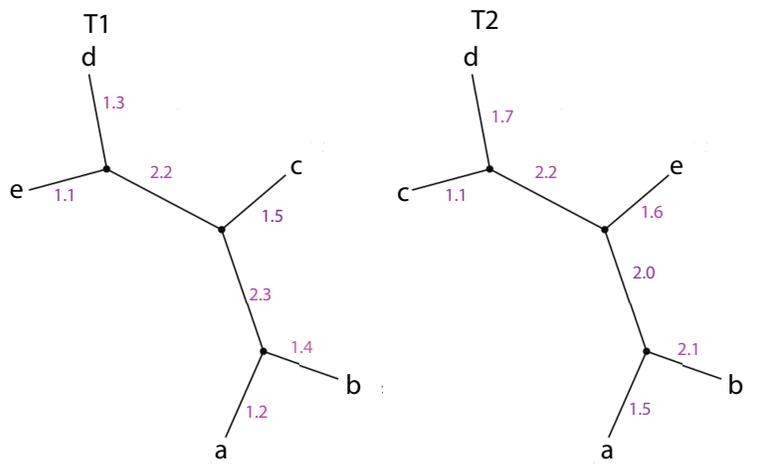}
\end{center}
\caption{Example phylogenetic trees: $T_1$ and $T_2$. The trees
  represent proposed most recent common ancestor relationships between
  5 taxa, labeled $a$ through $e$. These trees have branch lengths
  specified, but not all trees need have such information.}

\label{twotrees1}
\end{figure}

\subsection{Path difference}
The RF distance is completely determined by
the topologies of the trees, ignoring any edge lengths that may be
present. Conversely, the dissimilarity map distance requires
that the edge lengths be defined. The \emph{path difference} distance
$d_P$ is a distance analogous to the dissimilarity map, but which does
not require edge length information.

The calculation of the path difference is identical to the
dissimilarity map, except that elements in the distance matrix $D(T)$
are determined by counting the number of edges between the leaves,
rather than summing the edge lengths. (This is equivalent to the
dissimilarity map distance with all edge lengths in the tree set equal to $1$.)
The path difference is studied and compared with the RF distances by
\cite{Steel1993}.

Using the lexicographical ordering in the coordinates of the vector, we find
that the path difference vectorizations of our example trees are
\[
\begin{array}{ccc}
v(T_1)& =& (2,3,4,4,3,4,4,3,3,2), \\
 v(T_2)& = &(2,4,4,3,4,4,3,2,3,3). \\
\end{array}
\]
The path difference is therefore,
%$$ d_{p}(T_1,T_2) = \frac{1}{2} || v_{p}(T_1) - v_{p}(T_2) ||^2 = 3.$$
$ d_{p}(T_1,T_2) = || v(T_1) - v(T_2) ||_2 = \sqrt{6}$. % and the path
% difference is 
% $ d_{p}(T_1,T_2) = || v_{p}(T_1) - v_{p}(T_2) ||_2^2 = \sqrt{6}.$

\subsection{Edge difference}

This tree metric between two trees is defined by \cite{Clifford1971}.
% Suppose $v_e:\Tn \to \Z^{n\choose 2}$ is a function such that  the
% $(i, j)$th coordinate, where $1 \leq i < j \leq n$, of the $v_e(T)$ is
% the number of edges on the unique path  between leaves $i$ and $j$ on
% $T$.  
Suppose we have two trees $T_1, \, T_2 \in  \Tn$.  Then the {\em edge
  difference} $d_e$ is a distance measure between two trees $T_1, \, T_2 \in
\Tn$ such that 
\[
d_e(T_1, T_2) = ||v(T_1) - v(T_2)||_1.
\]

 The edge vectorization of any tree is exactly the same as
 the path difference vectorizations of the tree.  %Thus
% using the same vector ordering as in the dissimilarity map example, we find
% that the edge vectorizations of our example trees are
% \[
% \begin{array}{ccc}
% v_{e}(T_1)& =& (2,3,4,4,3,4,4,3,3,2), \\
%  v_{e}(T_2)& = &(2,4,4,3,4,4,3,2,3,3). \\
% \end{array}
% \]
The edge difference is therefore,
$ d_{e}(T_1,T_2) = || v(T_1) - v(T_2) ||_{1} = 6.$

\subsection{Precise $k$-interval cospeciaion}
The precise $k$-interval cospeciaion ($k$-IC) distance
$d_k$ is also a distance analogous to the path difference distance, but
which uses $l_{\infty}$ norm instead of $l_2$ norm.  This tree metric
was defined by \cite{cophy}.

The precise $k$-IC vectorization of any tree is exactly the same as
the path difference vectorizations of the tree.  % Thus
% Using the same vector ordering as in the dissimilarity map example, we find
% that the precise $k$-IC vectorizations of our example trees are
% \[
% \begin{array}{ccc}
% v_{k}(T_1)& =& (2,3,4,4,3,4,4,3,3,2), \\
%  v_{k}(T_2)& = &(2,4,4,3,4,4,3,2,3,3). \\
% \end{array}
% \]
The precise $k$-IC is therefore,
$ d_{k}(T_1,T_2) = || v(T_1) - v(T_2) ||_{\infty} = 1.$

Using the definitions of the tree differences $d_e, \, d_p, \, d_k$
between any two trees $T_1, \, T_2 \in  \Tn$  we can immediately have
the following remarks. 

\begin{remark}\label{rm:np}
\begin{itemize}
\item The tree differences $d_e, \, d_p, \, d_k$ between any two trees $T_1, \,
T_2 \in  \Tn$ are tree metrics.

\item The tree differences $d_e, \, d_p, \, d_k$ between any two trees $T_1, \,
T_2 \in  \Tn$  can be computed in $O(n^2)$.
\item Many tree metrics such as Nearest-Neighbor-Interchange distance,
  Subtree-Prune-and-Regraft distance, and
  Tree-Bisection-and-Regrafting distance are NP-hard \cite{Dasgupta97oncomputing,Hickey2008,Allen2001a}.
\end{itemize}
\end{remark}

\section{Distributions of the three tree metrics between
  unrooted random trees}\label{random_unrooted}
In this section we focus on the distributions of the path difference,
edge difference and precise $K$-IC distances between unrooted random
trees from $\Tn$. 

\subsection{Distribution of path difference metric between two
  trees}\label{path} 

%Let $\Tn$ be the set of all binary trees with $n$ leaves and
Suppose we sampled trees from the uniform distribution over $\Tn$.  In
this section we consider the distribution of the path difference tree
metric $d_p$ between two random trees sampled uniformly from $\Tn$.

Recall that $b(n)$ is the number of binary trees with $n$ labeled leaves. Then we have the following theorems.
\begin{theorem}[Theorem 3 from \cite{Steel1993}]\label{sd}
Consider the distribution of $d_p^2$ under the uniform distribution over
$\Tn$.  Let $d_{ij}(T)$ for $T \in \Tn$ be the number of edges
on the unique path between a leaf $i$ to a leaf $j$.  Then,
\begin{equation}
\begin{array}{cll}
\mathbf{E} [d_{ij}(T)]& =& \alpha(n),\\
\mathbf{V}[d_{ij}(T)] & = & 4n-6-\alpha(n)-\alpha^2(n),\\
\end{array} 
\end{equation}
where $\alpha(n+2) = \frac{2^{2^n}}{{2n \choose n}}$ and
\begin{equation}
\mu_p(n) = 2 {n \choose 2} \mathbf{V}[d_{ij}(T)]
\end{equation}
where $\mu_p(n)$ is the expected value of $d_p^2$ under the uniform distribution over
$\Tn$.
\end{theorem}

\begin{proof}
In this paper we only show the proof for $\mu_p(n)$.  The rest of the
proof for this theorem see \cite{Steel1993}.
By definition of $d_p^2$ we have:
\[
d_p^2(T,T') = \parallel  d(T) - d(T')  \parallel_2^2 = \sum\limits_{i<j} [d_{ij}(T) - d_{ij}(T')]^2,
\]
where $T$ and $T'$ are two random binary trees. So the mean is:
\begin{eqnarray}
\mu_p(n) & = & \mathbb{E} [d_p^2(T,T')] = \sum\limits_{T,T'} \Pr(T) \Pr(T') d_p^2(T,T') \nonumber\\
& = & \sum\limits_{T,T'} \frac{1}{b(n)^2} \sum\limits_{i<j} [d_{ij}(T) - d_{ij}(T')]^2 \nonumber\\
& = & \frac{1}{b(n)^2} \sum\limits_{T,T'} \sum\limits_{i<j} [d_{ij}(T)^2 + d_{ij}(T')^2 - 2 d_{ij}(T) d_{ij}(T')] \nonumber \\
& = & \frac{1}{b(n)^2} \sum\limits_{i<j} \left[\sum\limits_{T,T'} d_{ij}(T)^2 + \sum\limits_{T,T'} d_{ij}(T')^2 -2 \sum\limits_{T,T'} d_{ij}(T) d_{ij}(T') \right]  \nonumber \\
& = & \frac{1}{b(n)^2} \sum\limits_{i<j} \left[\sum\limits_{T'} \left(\sum\limits_{T} d_{ij}(T)^2\right) + \sum\limits_{T} \left(\sum\limits_{T'} d_{ij}(T')^2\right)  -2 \sum\limits_{T} d_{ij}(T) \left(\sum\limits_{T'}d_{ij}(T')\right) \right]  \nonumber \\
& = & \frac{1}{b(n)^2} \sum\limits_{i<j} \left[2 b(n) \sum\limits_{T} d_{ij}(T)^2  -2 \left(\sum\limits_{T}d_{ij}(T)\right)^2 \right].  \nonumber 
\end{eqnarray}
Notice that $\sum\limits_T f(d_{ij}(T))$ does not depend the selection of $i$ and $j$ because of the symmetry of labeling (it is easy to prove by contradiction and switching the labels). Therefore $\sum\limits_T f(d_{ij}(T)) = \sum\limits_T f(d_{kl}(T))$ with $i<j$, $k<l$, and thus we have:
\begin{eqnarray}
\mu_p(n) & = & \frac{2}{b(n)^2} \left( \begin{array}{c} n \\ 2 \end{array} \right)\left[ b(n) \sum\limits_{T} d_{ij}(T)^2  - \left(\sum\limits_{T}d_{ij}(T)\right)^2 \right]  \nonumber \\
& = & 2 \left( \begin{array}{c} n \\ 2 \end{array} \right)\left[ \sum\limits_{T} \frac{d_{ij}(T)^2}{b(n)}  - \left(\sum\limits_{T} \frac{d_{ij}(T)}{b(n)}\right)^2 \right]  \nonumber \\
& = & 2 \left( \begin{array}{c} n \\ 2 \end{array} \right)\left[ \sum\limits_{T} d_{ij}(T)^2 \Pr(T)  - \left(\sum\limits_{T} d_{ij}(T)\Pr(T)\right)^2 \right]  \nonumber \\
& = & 2 \left( \begin{array}{c} n \\ 2 \end{array} \right) \left( \mathbb E [d_{ij}(T)^2] - \mathbb E [d_{ij}(T)]^2 \right) = 2 \left( \begin{array}{c} n \\ 2 \end{array} \right) Var(d_{ij}(T)) \nonumber
\end{eqnarray}
with any selection of $i$ and $j$. 
\end{proof}

\begin{theorem}
$\sigma_p^2(n)$, the variance of $d_p^2$, is 
\[
\begin{array}{rcl}
\sigma_p^2(n) &=& \frac{1}{b(n)^2}  \left\{ \begin{split}
& \sum\limits_{T,T'} \left[ \sum\limits_{i<j} d_{ij}(T)^2\right]^2 + \sum\limits_{T,T'} \left[\sum\limits_{i<j} d_{ij}(T')^2\right]^2 + 4 \sum\limits_{T,T'} \left[\sum\limits_{i<j} d_{ij}(T) d_{ij}(T')\right]^2\\
+ & 2 \left\{\left( \begin{array}{c} n \\ 2 \end{array} \right) b(n) [4n-6-\alpha(n)]\right\}^2 \\
- & 8 b(n) \alpha(n)\sum\limits_{T} \left[ \sum\limits_{i<j} d_{ij}(T)^2\right]\left[\sum\limits_{i<j} d_{ij}(T)\right] 
\end{split}\right\} \\
&  & - 4\left[ {n \choose 2} \mathbf{V}[d_{ij}(T)]\right]^2.\\
\end{array}
\]
\end{theorem}

\begin{proof}

Since $\sigma_p^2(n) = Var (d_p^2) = \mathbb E [d_p^4] - \mu_p(n)^2$,
where the explicit formula of $\mu_p(n)$ is known, we have to consider
only $\mathbb E [d_p^4]$:
\begin{eqnarray}
\mathbb E [d_p^4(T,T')] & = & \sum\limits_{T,T'} \Pr(T) \Pr(T') [d_p^2(T,T')]^2 \nonumber\\
& = & \sum\limits_{T,T'} \frac{1}{b(n)^2} \left(\sum\limits_{i<j} [d_{ij}(T) - d_{ij}(T')]^2 \right)^2 \nonumber\\
& = & \frac{1}{b(n)^2} \sum\limits_{T,T'} \left[\sum\limits_{i<j} d_{ij}(T)^2 + \sum\limits_{i<j} d_{ij}(T')^2 - 2 \sum\limits_{i<j} d_{ij}(T) d_{ij}(T')\right]^2 \nonumber \\
& = & \frac{1}{b(n)^2}  \left\{ \begin{split}
& \sum\limits_{T,T'} \left[ \sum\limits_{i<j} d_{ij}(T)^2\right]^2 + \sum\limits_{T,T'} \left[\sum\limits_{i<j} d_{ij}(T')^2\right]^2 + 4 \sum\limits_{T,T'} \left[\sum\limits_{i<j} d_{ij}(T) d_{ij}(T')\right]^2\\
+ & 2 \sum\limits_{T,T'} \left[ \sum\limits_{i<j} d_{ij}(T)^2\right]\left[\sum\limits_{i<j} d_{ij}(T')^2\right] \\
- & 4 \sum\limits_{T,T'} \left[ \sum\limits_{i<j} d_{ij}(T)^2\right]\left[\sum\limits_{i<j} d_{ij}(T)d_{ij}(T')\right] \\
- & 4 \sum\limits_{T,T'} \left[ \sum\limits_{i<j} d_{ij}(T')^2\right]\left[\sum\limits_{i<j} d_{ij}(T)d_{ij}(T')\right]
\end{split}\right\} \nonumber \\
& = & \frac{1}{b(n)^2}  \left\{ \begin{split}
& \sum\limits_{T,T'} \left[ \sum\limits_{i<j} d_{ij}(T)^2\right]^2 + \sum\limits_{T,T'} \left[\sum\limits_{i<j} d_{ij}(T')^2\right]^2 + 4 \sum\limits_{T,T'} \left[\sum\limits_{i<j} d_{ij}(T) d_{ij}(T')\right]^2\\
+ & 2 \sum\limits_{T,T'} \left[ \sum\limits_{i<j} d_{ij}(T)^2\right]\left[\sum\limits_{i<j} d_{ij}(T')^2\right] \\
- & 8 \sum\limits_{T,T'} \left[ \sum\limits_{i<j} d_{ij}(T)^2\right]\left[\sum\limits_{i<j} d_{ij}(T)d_{ij}(T')\right] 
\end{split}\right\}. \nonumber 
\end{eqnarray}
In this equation, two terms can be simplified as:
\begin{eqnarray}
\sum\limits_{T,T'} \left[ \sum\limits_{i<j} d_{ij}(T)^2\right]\left[\sum\limits_{i<j} d_{ij}(T')^2\right]  & = & \left[\sum\limits_{T}  \sum\limits_{i<j} d_{ij}(T)^2\right]\left[\sum\limits_{T'} \sum\limits_{i<j} d_{ij}(T')^2\right] \nonumber \\
& = &  \left[\left( \begin{array}{c} n \\ 2 \end{array} \right) \sum\limits_{T}  d_{ij}(T)^2\right]^2 \nonumber \\
& = &  \left\{\left( \begin{array}{c} n \\ 2 \end{array} \right) b(n) \mathbb E [d_{ij}(T)^2]\right\}^2 \nonumber \\
& = & \left\{\left( \begin{array}{c} n \\ 2 \end{array} \right) b(n) [4n-6-\alpha(n)]\right\}^2. \nonumber
\end{eqnarray}

\begin{eqnarray}
\sum\limits_{T,T'} \left[ \sum\limits_{i<j} d_{ij}(T)^2\right]\left[\sum\limits_{i<j} d_{ij}(T)d_{ij}(T')\right]  & = & \sum\limits_{T} \left[ \sum\limits_{i<j} d_{ij}(T)^2\right]\left[\sum\limits_{T'}\sum\limits_{i<j} d_{ij}(T)d_{ij}(T')\right] \nonumber \\
& = &\sum\limits_{T} \left[ \sum\limits_{i<j} d_{ij}(T)^2\right]\left[\sum\limits_{i<j} d_{ij}(T) b(n) \mathbb E [d_{ij}(T)]\right] \nonumber \\
& = & b(n) \alpha(n)\sum\limits_{T} \left[ \sum\limits_{i<j} d_{ij}(T)^2\right]\left[\sum\limits_{i<j} d_{ij}(T)\right]. \nonumber
\end{eqnarray}

\end{proof}

\subsection{Distribution of the edge difference metric between two trees}\label{edge}

\begin{theorem}\label{kic2}
Consider the distribution of $d_e$ under the uniform distribution over
$\Tn$. Then, using the relation between $l_p$ norm and $l_q$ norms
where $0 < q < p$ such that $||x||_p \leq ||x||_q \leq m^{(\frac{1}{q}
- \frac{1}{p})}$, we have the following theorem:
\begin{equation}
\sqrt{2 {n \choose 2} \left( 4n-6-\alpha(n)-\alpha^2(n)\right)} \leq \mu_e(n) \leq {n \choose 2}  \sqrt{2 \left( 4n-6-\alpha(n)-\alpha^2(n)\right)}
\end{equation}
where  $\mu_e(n)$ is the expected value of $d_e$ under the uniform distribution over
$\Tn$.
\end{theorem}

% \begin{theorem}\label{mae}
% Consider the distribution of $d_e$ under the uniform distribution over
% $\Tn$. Then, 
% \begin{equation}
% \mu_e(n) = 2 \cdot {n \choose 2}\cdot \mbox{MAE}( d_{ij}(T)),
% \end{equation}
% where  $\mu_e(n)$ is the expected value of $d_e$ under the uniform distribution over
% $\Tn$ and $\mbox{MAE} (d_{ij}(T))$ is the mean absolute error of $d_{ij}(T)$.
% \end{theorem}
% \begin{proof}
% We have
% \[
% \mu_e(n) = \frac{\sum_{T, T' \in \Tn}  \sum_{i < j} |d_{ij}(T) - d_{ij}(T')|}{b(n)}
% \]
% where $b(n)$ is the number of all unrooted binary trees.  Now we
% interchange the sums.  Then we have
% \[
% \mu_e(n) = \frac{\sum_{i < j} \sum_{T, T' \in \Tn}  |d_{ij}(T) - d_{ij}(T')|}{b(n)}.
% \]
% Since 
% $\sum_{T, T' \in \Tn}  |d_{ij}(T) - d_{ij}(T')|$ has the same value
% for all possible $i < j$, we have
% \[
% \mu_e(n) = \frac{ {n \choose 2} \sum_{T, T' \in \Tn}  |d_{ij}(T) -
%   d_{ij}(T')|}{b(n)} = 2 \cdot {n \choose 2} \cdot \mbox{MAE}(d_{ij}(T)).
% \]
% \end{proof}
% From Theorem \ref{sd} and Theorem \ref{mae} we can see that if we
% think there are outliers in the data set we should use $d_e$ instead
% of using $d_p$ since the mean-absolute error is more robust than the
% standard deviation to outliers \cite{MAE}.

\begin{remark}
Let $B(x) = \sum_{n > 0} \frac{b(n+1)}{n!} x^n$ be an exponential
generating function for the number of planted binary trees, $b(n+1)$,
with $n$ labeled non-root leaves (or the number of rooted binary trees
with $n$ leaves).  Let 
\[
F(x, y) = y B(x) + y^2 B(x) + \ldots = \frac{1}{[1 - yB(x)]} - 1
\]
be the exponential generating function for the number of ordered
forests consisting of a given number of rooted trees (marked by $y$)
and a given number of leaves (marked by $x$).  Then for a fixed pair
of distinct leaves $i$ and $j$ (we can set $i = 1$ and $j = 2)$, we have
\[
\sum_{T \in \Tn} \sum_{T' \in \Tn}|d_{ij}(T) - d_{ij}(T')| =
\sum_{r = 2}^{n-1} [y^r][x^{n-2}] yF(x, y) \left(\sum_{r' = 2}^{n-1}
  |r - r'| [y^{r'}][x^{n-2}] y F(x, y) \right),   
\]
where
$[x^k][y^{k'}]f(x, y)$ denotes the coefficient of $x^k\cdot y^{k'}$ in the function $f(x, y)$.

\end{remark}

\subsection{Distribution of the precise $k$-IC tree metric between two trees}\label{kic_sec}
Now we consider the distribution of $d_k$ under the uniform distribution over
$\Tn$. Then, using the relation between $l_p$ norm and $l_q$ norms
where $0 < q < p$ such that $||x||_p \leq ||x||_q \leq m^{(\frac{1}{q}
- \frac{1}{p})}$, we have the following theorem:
\begin{theorem}\label{kic1}
Consider the distribution of $d_k$ under the uniform distribution over
$\Tn$. Then, 
\begin{equation}
\sqrt{2 \left( 4n-6-\alpha(n)-\alpha^2(n)\right)} \leq \mu_k(n) \leq \sqrt{2 {n \choose 2} \left( 4n-6-\alpha(n)-\alpha^2(n)\right)}
\end{equation}
where  $\mu_k(n)$ is the expected value of $d_k$ under the uniform distribution over
$\Tn$. 
\end{theorem}

\begin{remark}\label{bounds}
Using the same relation above, we can use $\mu_k(n)$ as an upper
bound for $\sqrt{\mu_p(n)}$ and $\mu_e(n)$, that is
\[
\begin{array}{ccc}
\sqrt{\mu_p(n)} &\leq & \sqrt{{n \choose 2}} \mu_k(n)\\
\mu_e(n) &\leq & {n \choose 2}\mu_k(n).\\
\end{array}
\]
\end{remark}

\section{Species tree and gene tree under the coalescent}\label{sp_gene}

Let $\Tn'$ be the space of rooted trees with $n$ leaves.  Note that
$\Tn' = {\mathcal T}_{n+1}$.
In this section we consider the distances between a species tree and a
gene tree under the coalescent given the species tree.  
First we consider the following two lemmas from \cite[]{Coons}.
\begin{lemma}[Lemma 1 from \cite{Coons}]\label{coon1}
For any two trees $T_1, \, T_2 \in \Tn'$, $d_{k}(T_1, \, T_2) \leq
(n-2).$ 
\end{lemma}

A {\em caterpillar tree} is any unrooted binary phylogenetic tree which
reduces to the path if we delete all edges attached to a leaf and all
leaves (see Figure \ref{coalesced} for an example).

\begin{lemma}[Corollary 1 from \cite{Coons}]\label{coon2}
If $d_{k}(T_1, \, T_2) = (n-2) $ for  $T_1, \, T_2 \in \Tn'$, then
$T_1$ or $T_2$ is a caterpillar tree.
\end{lemma}

\cite{Coons} considered unrooted trees in $\Tn$. In the case of
unrooted trees in $\Tn$, we have the bound $(n - 3)$ in Lemma
\ref{coon1} and Lemma \ref{coon2}.   But in this section
we consider $\Tn'$, the space of rooted trees and using the fact that
$\Tn' = \mathcal{T}_{n+1}$, thus we have the bound $((n+1) - 3) = (n -
2)$.  For example, if we consider $T_1$ and $T_2 $ in $\Tn'$ as seen
Figure \ref{twotrees2}, then $d_k(T_1, \, T_2) = ||(2, 3, 3) - (3, 2,
3) ||_{\infty} =  (3 - 2) = 1$.

\begin{figure}
\begin{center}
\includegraphics[width=4in]{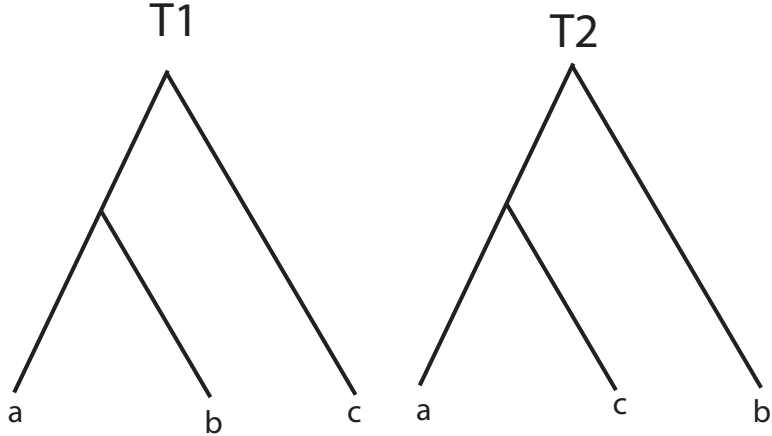}
\end{center}
\caption{Example phylogenetic rooted trees: $T_1$ and $T_2$. The trees
  represent proposed most recent common ancestor relationships between
  3 taxa, labeled $a$ through $c$.}
\label{twotrees2}
\end{figure}

Thus, a caterpillar tree is a special case, so we consider that the
species tree $T_s \in \Tn'$ be a caterpillar tree.  In this section
we also consider a sample size of individuals from each species is one
and each species has the same effective population size $N_e$.  Let
$t_i$ be a time interval in the coalescent time unit between the $(i -
1)$th event when two species are coalesced to the $i$th event when two
species are coalesced (see figure \ref{coalesced}).  

\begin{figure}
\begin{center}
\includegraphics[width=3in]{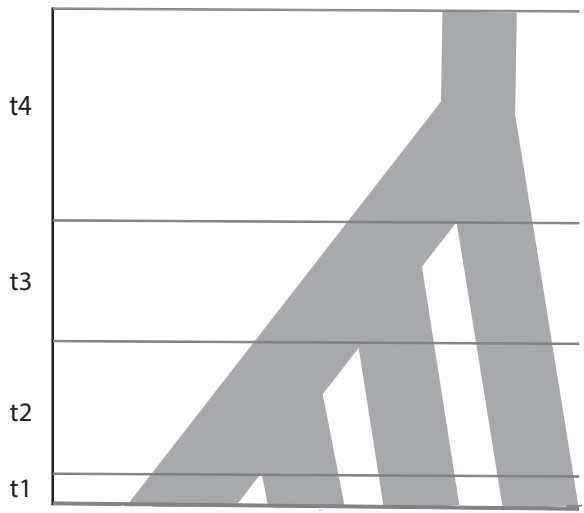}
\end{center}
\caption{The caterpillar species tree $T_s$ with $n = 4$.  }\label{coalesced}
\end{figure}

% Now we consider the probability that the species tree $T_s \in \Tn'$
% which is caterpillar and a gene tree $T_g$ generated by the coalescent given
% the species tree $T_s$.
Let $T_s \in \Tn'$ be  a caterpillar tree.  
Now we consider the probability that $T_s \in \Tn'$
and a gene tree $T_g$ generated by the coalescent given
the species tree $T_s$ have the same tree topology.

Let $g_{ij}(t)$ be the probability
that $i$ lineages derive from $j$ lineages that existed $t > 0$
coalescent time units in the past such that
\[
g_{ij}(t) = \sum_{k = j}^i \exp(\frac{-k(k-1)t}{2})\frac{(2k-1)(-1)^{k-j}j_{(k-1)}i_{[k]}}{j!(k-j)!i_{(k)}},
\]
where $a_{(k)}= a(a+1)\ldots (a+k-1)$ for $k \geq 1$ with $a_{(0)} =
1$; and $a_{[k]} = a(a-1)\ldots (a-k+1)$ for $k \geq 1$ with $a_{(0)} =
1$ \cite[]{Takahata1989,Takahata,Tavare}.  $g_{ij}(t) = 0$ except with
$1 \leq j \leq i$. 

\begin{remark}
If $t$ is a scale of coalescent time units then $t$ can be written as
$t = \frac{t'}{N_e}$ where $t'$ is the number of generation and $N_e$ is a
population size.  We assume that the
size of an ancestral species is the sum of the sizes of its
descendants so that the scaling of time would be different
before and after the divergence of the ancestor, i.e., before
diverging the scale of coalescent time unit would be $t =
\frac{t'}{2N_e}$ and after diverging it would be $t =
\frac{t'}{N_e}$.
\end{remark}

\begin{remark}
In fact, we can simplify $g_{21}(t_i)$ for some coalescent time
interval $t_i > 0$ and it can be written as 
\[
 g_{21}(t_i) = 1-\exp(-t_i).
\]
\end{remark}

Before we show the probability that any of these three distribution
between the caterpillar species tree and gene trees generated from the
coalescent process equals to zero, we have to define some notation.  

To consider this problem, we need to count the number of cases of $M
\in \N$ branches with $N \in \N$ lineages in total. Let $C_{N,M}$ be the number of
cases that $N$ lineages coalesce to $M$ lineages. We call the number of lineages in a specific branch the ``branch degree". Obviously, the answer depends on if we consider the orders among branches with the same branch degree. If we consider the two figures in Figure \ref{fig:order} as different cases, then it is not very difficult to obtain that $C_{N,M} = \displaystyle \frac{\prod_{i=2}^N {i \choose 2}}{\prod_{i=2}^M {i \choose 2}}$. However, it will be more complicate if we consider them as the same case. We need to first enumerate all possible ordered $M$ branch degrees (number of lineages coalesce in the branch), then sum up the number of cases for each ordered branch degrees. For example, when $N=5$ and $M=3$, we have two possible ordered branch degrees $(113)$ and $(122)$; since for we have ${5 \choose 3}*(2\cdot 3-3)!!=30$ cases for $(113)$, and ${5 \choose 2}{3 \choose 2}/2=15$ cases for $(122)$, we have 45 cases in total.

\begin{figure}[ht]
  \centering
  \subfigure[12 happens after 34]{
    \label{fig:order1}     %% label for 1st subfigure
    \includegraphics[width=3in]{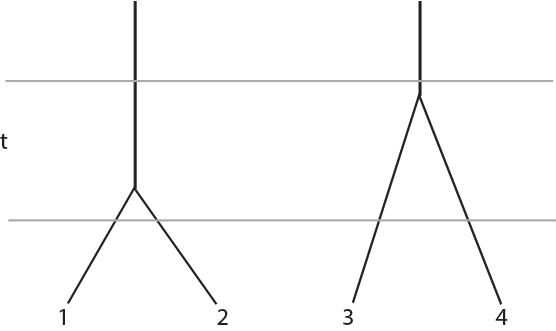}}
  \hspace{0.5in}
  \subfigure[34 happens after 12]{
    \label{fig:order2}     %% label for 2nd subfigure
    \includegraphics[width=3in]{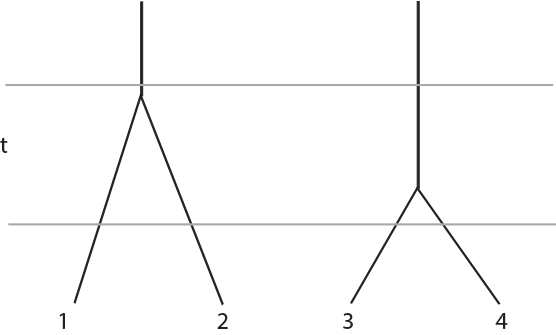}}
  \caption{4 lineages coalesce to 2 lineages with the same topology $12\mid 34$}
  \label{fig:order}     %% label for entire figure
{\footnotesize }
\end{figure}

Define $\mathcal D_{M,N}=\{(w_1, w_2, \ldots, w_M)\in \mathbb Z_+^M: \sum\limits_{l=1}^M w_l=N,\ w_1\leq w_2 \leq \ldots \leq w_M\}$ as the set of all possible ordered branch degrees. It is trivial to prove that we can enumerate all elements in $\mathcal D_{M,N}$ without duplication in the following way:
\begin{eqnarray}
\mathcal D_{M,N}=\large\{(w_1, w_2, \ldots, w_M)\in \mathbb Z_+^M: & &
w_1 = 1, 2, \ldots, \lfloor \frac{N}{M} \rfloor, \ w_2 = 1, 2, \ldots, \lfloor \frac{N-w_1}{M-1} \rfloor, \nonumber\\
& \ldots, & w_{M-1} = 1, 2, \ldots, \lfloor \frac{N-\sum\limits_{l=1}^{M-2} w_l}{M-1} \rfloor,  \ w_M = N- \sum\limits_{l=1}^{M-1} w_l
\large\}, \nonumber
\end{eqnarray}
where ``$\lfloor \cdot \rfloor$" gives the largest integer that is smaller than a specific real number. We can define an 1-1 mapping over $\mathcal D_{M,N}$ such that $\forall \mathbf w=(w_1, w_2, \ldots, w_M) \in \mathcal D_{M,N}$, $\mathbf w$ maps to two vectors $\mathbf n(\mathbf w) = (n_0, n_1, \ldots, n_l)\in \mathbb Z_+^{l+1}$ and $\mathbf u(\mathbf w)=(u_0,u_1, \ldots, u_l) \in \mathbb Z_+^{l+1}$ which satisfy
\[
\mathbf w = (\underbrace{n_0, \cdots, n_0}_{u_0 \ many}, \ \underbrace{n_1, \cdots, n_1}_{u_1 \ many}, \cdots, \underbrace{n_l, \cdots, n_l}_{u_l  \ many}).
\]
where $n_0=1<n_1<n_2<\cdots<n_l$. Notice that this implies $\sum\limits_{\alpha=0}^{l} u_\alpha = M$ and $\sum\limits_{\alpha=0}^{l} u_\alpha n_\alpha = N$.

\begin{lemma}\label{lem:CMN}
\[
C_{M,N} = \sum\limits_{\mathbf n(\mathbf w), \mathbf u(\mathbf w): \mathbf w \in \mathcal D_{M,N}}  \left\{ \frac{N!}{u_0!} \prod\limits_{\alpha=1}^{l} \frac{((2n_\alpha-3)!!)^{u_\alpha}}{u_\alpha! (n_\alpha!)^{u_\alpha}} \right\}.
\]
\end{lemma}

\begin{proof}
Consider $\mathbf n(\mathbf w)$ and $\mathbf u(\mathbf w)$ of an arbitrary $\mathbf w \in \mathcal D_{M,N}$. We have $u_\alpha$ branches with degree $n_\alpha$, $\alpha = 0, 1, \ldots, l$. For each branch with degree $n_\alpha$, we have $(2n_\alpha-3)!!$ different tree topologies. Notice that we don't consider the permutation among the $u_\alpha$ branches with degree $n_\alpha$. Thus the number of cases that we choose first $u_1$ branches with degree $n_1$ is:
\begin{eqnarray}
& & \frac{{N \choose n_1}{N-n_1 \choose n_1}\cdots {N-(u_1-1)n_1 \choose n_1} [(2n_1-3)!!]^{u_1}}{u_1!} \nonumber\\
& = & \frac{N!}{n_1!\text{\sout{$(N-n_1)!$}}} \cdot \frac{\text{\sout{$(N-n_1)!$}}}{n_1!\text{\sout{$(N-2n_1)!$}}} \cdots \frac{\text{\sout{$(N-(u_1-1)n_1)!$}}}{n_1!(N-u_1n_1)!} \cdot \frac{[(2n_1-3)!!]^{u_1}}{u_1!}\nonumber \\
& = & \frac{N!}{(n_1!)^{u_1} (N-u_1n_1)!} \frac{[(2n_1-3)!!]^{u_1}}{u_1!}. \nonumber
\end{eqnarray}

\end{proof}
Therefore, consider the rest branches, the total number of cases, $C_{M,N}$, is:
\begin{eqnarray}
&  & \frac{{N \choose n_1}\cdots {N-(u_1-1)n_1 \choose n_1} [(2n_1-3)!!]^{u_1}}{u_1!}  \cdot
\frac{{N-u_1 n_1 \choose n_2}\cdots {N-u_1 n_1-(u_2-1)n_2 \choose n_2} [(2n_2-3)!!]^{u_2}}{u_2!}   \nonumber\\
&  & \cdots \frac{{N-\sum\limits_{\alpha=1}^{l-1} u_\alpha n_\alpha \choose n_l}\cdots {N-\sum\limits_{\alpha=1}^{l-1} u_\alpha n_\alpha -(u_l-1)n_l \choose n_l} [(2n_l-3)!!]^{u_l}}{u_l!}  \nonumber\\
& = & \frac{N!}{(n_1!)^{u_1} \text{\sout{$(N-u_1n_1)!$}}} \frac{[(2n_1-3)!!]^{u_1}}{u_1!} \cdot
\frac{ \text{\sout{$(N-u_1n_1)!$}}}{(n_2!)^{u_2}  \text{\sout{$(N-\sum\limits_{\alpha=1}^{2} u_\alpha n_\alpha)!$}}} \frac{[(2n_2-3)!!]^{u_2}}{u_2!}  \nonumber\\
&  & \cdots \frac{\text{\sout{$(N-\sum\limits_{\alpha=1}^{l-1} u_\alpha n_\alpha)!$}}}{(n_l!)^{u_l} (N-\sum\limits_{\alpha=1}^{l} u_\alpha n_\alpha)!} \frac{[(2n_l-3)!!]^{u_l}}{u_l!} \nonumber \\
& = &  N! \cdot \frac{[(2n_1-3)!!]^{u_1}}{(n_1!)^{u_1} u_1!} \cdot
\frac{[(2n_2-3)!!]^{u_2}}{(n_2!)^{u_2} u_2!}  \cdots 
\frac{[(2n_l-3)!!]^{u_l}}{(n_l!)^{u_l} u_l!} \frac{1}{ (u_0)!}. \nonumber
\end{eqnarray}

\begin{example}
The following table gives the values of $C_{M,N}$ when $N\leq 6$:
\begin{center}
\begin{tabular}{l|*{7}{c}}
\backslashbox{$M$}{$N$} & 1 & 2 & 3 & 4 & 5 & 6 \\\hline
1 & 1 & 1 & 3 & 15 & 105 & 945 \\
2 &   & 1 & 3 & 15 & 105 & 945  \\
3 &   &   & 1 & 6 & 45 & 420 \\
4 &   &   &   & 1 & 10 & 105  \\
5 &   &   &   &   & 1 & 15  \\
6 &   &   &   &   &   & 1 
\end{tabular}
\end{center}
Take $N=6$, $M=3$ for example. There are 3 possible ordered branch degrees:
\begin{enumerate}
\item
$\mathbf w = (114)$, $\mathbf n = (14)$, $\mathbf u = (21)$, number of cases: $\frac{6!}{2!}\cdot \frac{[(2*4-3)!!]^1}{1!(4!)^1}=225$;
\item
$\mathbf w = (123)$, $\mathbf n = (123)$, $\mathbf u = (111)$, number of cases: $\frac{6!}{1!}\cdot \frac{[(2*2-3)!!]^1}{1!(2!)^1}\cdot \frac{[(2*3-3)!!]^1}{1!(3!)^1}=180$;
\item
$\mathbf w = (222)$, $\mathbf n = (12)$, $\mathbf u = (03)$, number of cases: $\frac{6!}{0!}\cdot \frac{[(2*2-3)!!]^3}{3!(2!)^3}=15$.
\end{enumerate}
So $C_{6,3} = 225+180+15=420$.
\end{example}

% \begin{remark}
% By Lemma \ref{lem:CMN} it is easy to show that $C_{N,1} = (2N-3)!!$ and $C_{N,N-1} = {N \choose 2}$. But it is not clear to show $C_{N,2} = (2N-3)!!$ (Equation (1) in \url{http://arxiv.org/pdf/0906.1317v1.pdf} may be useful, but it's still not very clear to me...). There is a possibility of further simplifying Lemma \ref{lem:CMN}.
% \end{remark}

For $n$ species, $n-1$ coalescences should happen during coalescent
times $t_1, t_2, \ldots, t_n$. Here, we call the pattern of how these
coalescences (regardless of which lineages are coalescing) distributed
over the coalescent times, i.e. in which coalescent time does the
$k_{th}$ coalescent happen, the coalescent timeline. When the gene
tree completely matches the species tree, we know that the tree
topology of the gene tree is fixed, i.e. the pattern and ordering of
coalescent are fixed. This means that the only thing we need to think
about is the coalescent timeline. Let's first see a simple example. 

Recall: $g_{ij}(t)$ is the probability that $i$ lineages coalesce to j lineages in time $t$.
\begin{example}
Consider 3 species. Fix the species tree to be $12\mid 3$. Figure \ref{fig:size3ex} gives all possible gene trees based on this species tree.
\begin{figure}[ht]
  \centering
  \subfigure[$12\mid 3$]{
    \label{fig:3a}     %% label for 1st subfigure
    \includegraphics[width=2in]{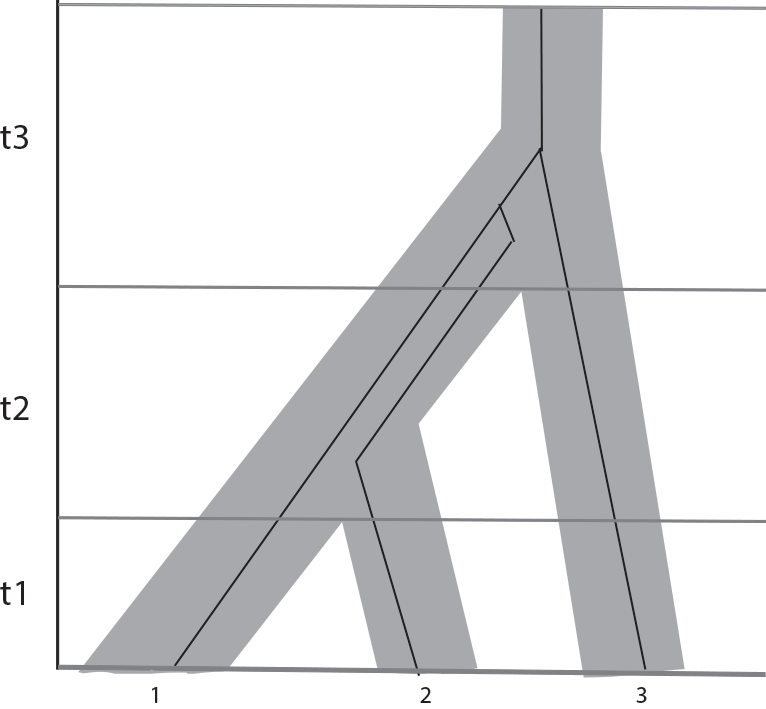}}
  \hspace{0.5in}
  \subfigure[$13\mid 2$]{
    \label{fig:3b}     %% label for 2nd subfigure
    \includegraphics[width=2in]{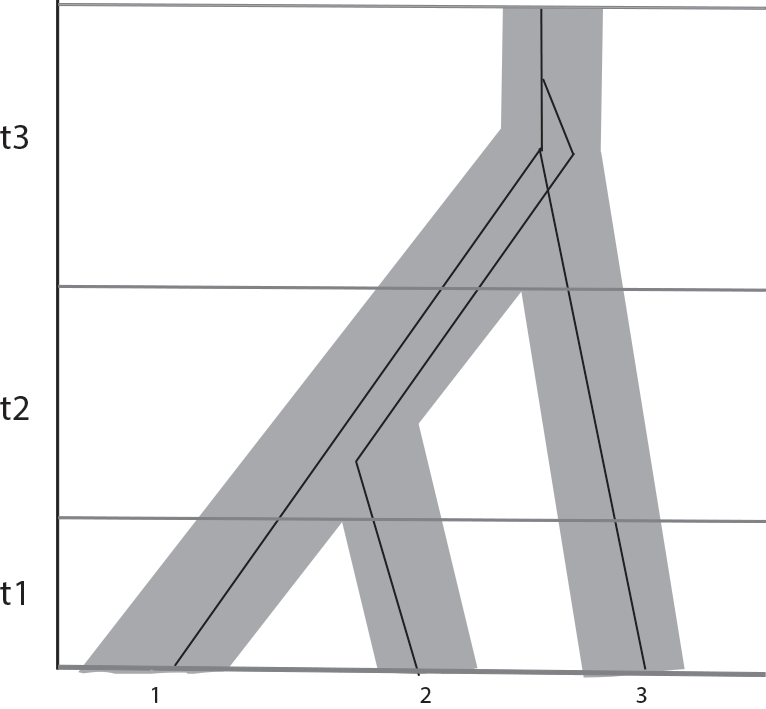}}
      \hspace{0.5in}
  \subfigure[$1\mid 23$]{
    \label{fig:3c}     %% label for 3rd subfigure
    \includegraphics[width=2in]{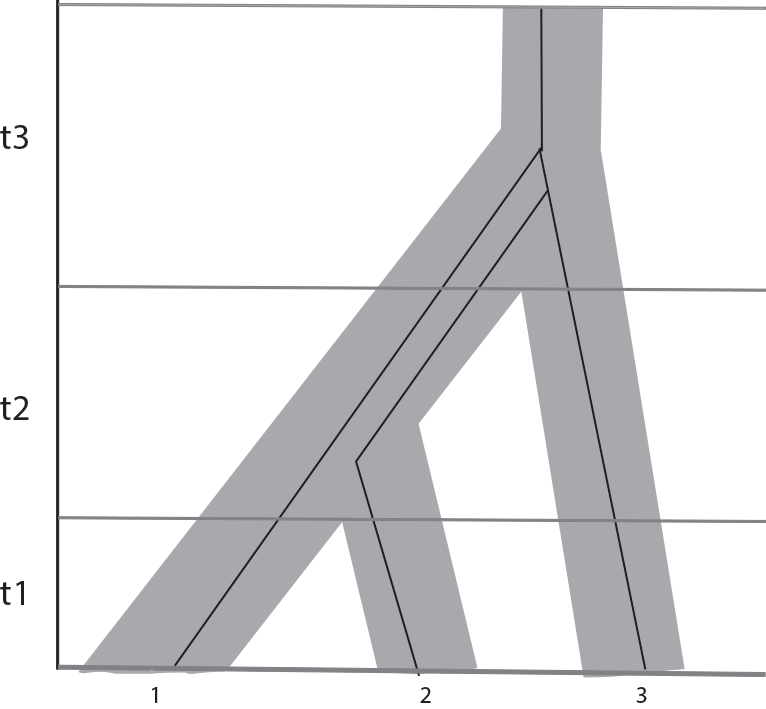}}
  \hspace{0.5in}
  \subfigure[$12\mid 3$]{
    \label{fig:3d}     %% label for 4th subfigure
    \includegraphics[width=2in]{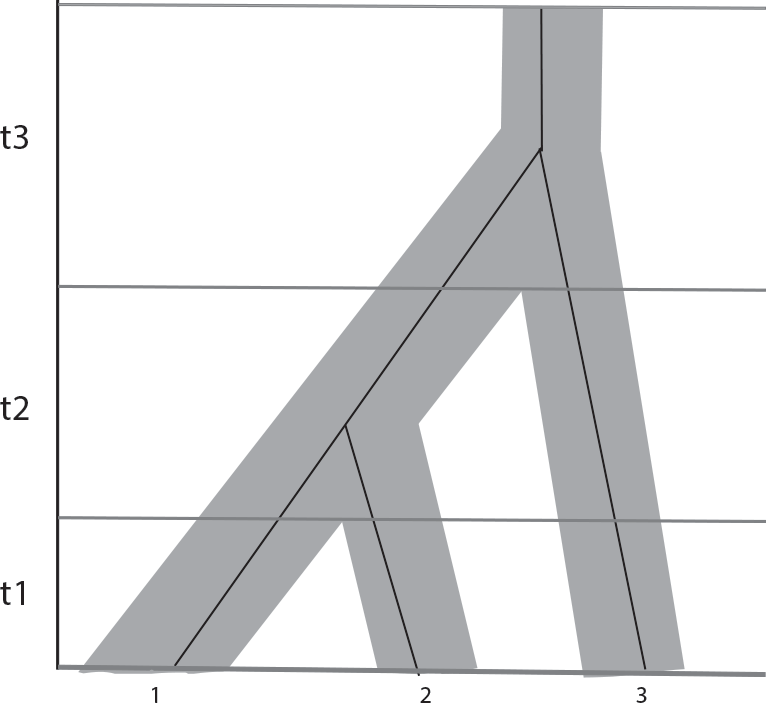}}
  \caption{All possible gene trees for the fixed species tree $12\mid 3$}
  \label{fig:size3ex}     %% label for entire figure
{\footnotesize }
\end{figure}

We can compute the probabilities of these trees as following and verify them by summing up to 1:
\begin{itemize}
\item
Cases for Figure \ref{fig:3a}, Figure \ref{fig:3b} and Figure \ref{fig:3c}: 
\begin{eqnarray}
& & \Pr((1,2) \mbox{ in } t_3, (12, 3) \mbox{ in } t_3) =\Pr((1,3) \mbox{ in } t_3, (13, 2) \mbox{ in } t_3) \nonumber\\
&= & \Pr((2,3) \mbox{ in } t_3, (23, 1) \mbox{ in } t_3)=\frac 1 {C_{3,1}}g_{22}(t_2)=\frac13 e^{-t_2}. \nonumber
\end{eqnarray}
 Notice that we have $\frac 1 {C_{3,1}}$ here because all these trees share the same coalescent timeline (both coalescences happen in $t_3$), and we have $C_{3,1}$ cases in $t_3$ where 3 lineages coalesce to 1 lineage;

%$\Pr(\text{Figure \ref{fig:3a}})=\Pr(\text{Figure \ref{fig:3b}})=\Pr(\text{Figure \ref{fig:3c}})=\frac 1 {C_{3,1}}g_{22}(t_2)=\frac13 e^{-t_2}$, notice that we have $\frac 1 {C_{3,1}}$ here because all these trees share the same coalescent timeline (both coalescences happen in $t_3$), and we have $C_{3,1}$ cases in $t_3$ where 3 lineages coalesce to 1 lineage;
\item
Case for Figure \ref{fig:3d}: $\Pr((1,2) \mbox{ in } t_2, (12, 3) \mbox{ in } t_3)=g_{21}(t_2)=1-e^{-t_2}$.
\end{itemize}
In this example, $\Pr(d(T_s, T_e)=0)=\Pr((1,2) \mbox{ in } t_3, (12, 3) \mbox{ in } t_3)+\Pr((1,2) \mbox{ in } t_2, (12, 3) \mbox{ in } t_3)=1-\frac23 e^{-t_2}$.
\end{example}

Since for each coalescent timeline, there is only one case gives a
gene tree which completely matches the species tree, all we need to do
is enumerate the coalescent timeline and compute probability for each
of them. 

\begin{theorem}\label{thm:probd0}
For $n$ species,
\begin{eqnarray}\label{eqn:probd0}
\Pr(d(T_s, T_e)=0) = \sum\limits_{i_2=0}^1 \sum\limits_{i_3=i_2}^2 \cdots \sum\limits_{i_k=i_{k-1}}^{k-1} \cdots \sum\limits_{i_{n-1}=i_{n-2}}^{n-2}
\left\{
\left[\prod\limits_{k=2}^{n-1} \frac{g_{k-i_{k-1}, k-i_k} (t_k)}{C_{k-i_{k-1}, k-i_k}}\right] \cdot \frac 1 {C_{n-i_{n-1},1}}
\right\},
\end{eqnarray}
where $i_1=0$.
\end{theorem}

\begin{proof}
Several requirements when we enumerate the coalescent timelines: 1) no coalescent in time $t_1$; 2) if the $i_{th}$ coalescence happens in time $t_{k_i}$, then $i+1\leq k_i\geq n$; 3) if the $i_{th}$ and $j_{th}$ coalescences happen in time $t_{k_i}$ and $t_{k_j}$ respectively and $i<j$, then $k_i \leq k_j$ (otherwise the gene tree will have a different tree topology with the species tree); 4) all lineages coalescent to one in time $t_n$.

In Equation \ref{eqn:probd0}, every choice of $(i_1, i_2, \ldots, i_{n-1})$ gives a possible coalescent timeline: $i_k$ coalescences happen before or during time $t_k$, $k=1, 2, \ldots, n-1$, and $(n-i_{n-1})$ coalescences happen during time $t_n$. It is trivial to see the these choices enumerate all possible coalescent timelines without duplicate.

Now consider a specific $(i_1, i_2, \ldots, i_{n-1})$. Then during time $t_k$, $k=2, 3, \ldots, n-1$, since the input has $k$ species with $i_{k-1}$ coalescences, i.e. $k-i_{k-1}$ lineages, and the output has $k$ species with $i_k$ coalescences, i.e. $k-i_k$ lineages, the probability that the gene tree completely agree with the species tree is $\frac{g_{k-i_{k-1}, k-i_k} (t_k)}{C_{k-i_{k-1}, k-i_k}}$ (see example in Figure \ref{fig:probd0proof}). During time $t_n$, we left $n-i_{n-1}$ lineages and they should coalesce to one, so the probability should be $\frac 1 {C_{n-i_{n-1},1}}$.

\begin{figure}[ht]
  \centering
    \includegraphics[width=3in]{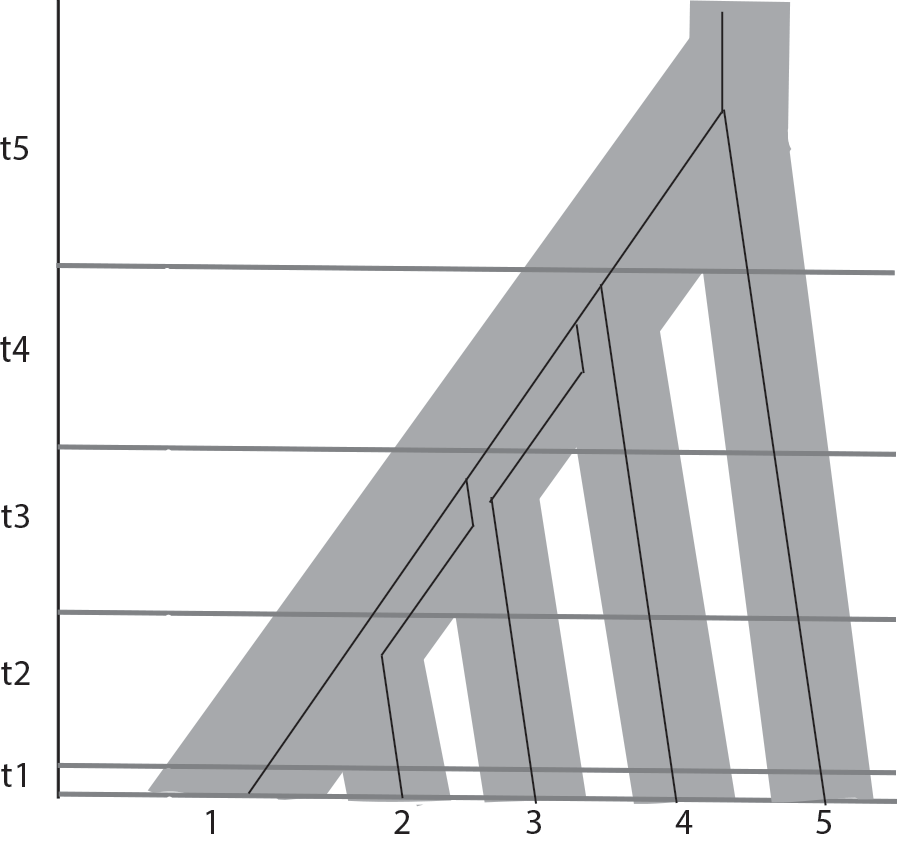}
  \caption{5 species with timeline: $(i_1, i_2, i_3, i_4) = (0,0,1,3)$}
  \label{fig:probd0proof}     %% label for entire figure
{\footnotesize $i_2-i_1=0$ coalescent happened in $t_2$; $i_3-i_2=1$ coalescent happened in $t_3$; $i_4-i_3=2$ coalescents happened in $t_4$.\\
In time $t_3$, we have $3-i_2=3$ lineages coming and $3-i_3=2$ lineages coming out, so the \\
probability that we get exactly the same topology as this figure during time $t_3$ is $\frac{g_{32}(t_3)}{C_{3,2}}$.}
\end{figure}

\end{proof}

\begin{example}
There are five cases for $n=4$ so that gene tree completely matches the species tree. We apply Theorem \ref{thm:probd0} for $n=4$ in and obtain the following probabilities for each of the cases:
\begin{enumerate}
\item
Coalescents $(1,2)$ in $t_4$; $(12, 3)$ in $t_4$; $(123, 4)$ in $t_4$ (see Figure \ref{fig:4a}). Probability is $\frac 1{15} g_{22}(t_2)g_{33}(t_3)$;
\item
Coalescents $(1,2)$ in $t_3$; $(12, 3)$ in $t_4$; $(123, 4)$ in $t_4$ (see Figure \ref{fig:4b}). Probability is $\frac 19 g_{22}(t_2)g_{32}(t_3)$;
\item
Coalescents $(1,2)$ in $t_3$; $(12, 3)$ in $t_3$; $(123, 4)$ in $t_4$ (see Figure \ref{fig:4c}). Probability is $\frac 13 g_{22}(t_2)g_{31}(t_3)$;
\item
Coalescents $(1,2)$ in $t_2$; $(12, 3)$ in $t_4$; $(123, 4)$ in $t_4$ (see Figure \ref{fig:4d}). Probability is $\frac 13 g_{21}(t_2)g_{22}(t_3)$;
\item
Coalescents $(1,2)$ in $t_2$; $(12, 3)$ in $t_3$; $(123, 4)$ in $t_4$ (see Figure \ref{fig:4e}). Probability is $g_{21}(t_2)g_{21}(t_3)$;
\end{enumerate} 
%We apply Theorem \ref{thm:probd0} for $n=4$ in Figure \ref{fig:size4ex} and obtain the following formula:
Then we have formula:
\begin{eqnarray}
\Pr(d(T_s, T_e)=0) = \frac 1{15} g_{22}(t_2)g_{33}(t_3) + \frac 19 g_{22}(t_2)g_{32}(t_3) + \frac 13 g_{22}(t_2)g_{31}(t_3) + \frac 13 g_{21}(t_2)g_{22}(t_3) + g_{21}(t_2)g_{21}(t_3).\nonumber
\end{eqnarray}
%\begin{eqnarray}
%\Pr(d(T_s, T_e)=0) & = & \Pr(\text{Figure \ref{fig:4a}}) + \Pr(\text{Figure \ref{fig:4b}}) +\Pr(\text{Figure \ref{fig:4c}})  + \Pr(\text{Figure \ref{fig:4d}})  + \Pr(\text{Figure \ref{fig:4e}})  \nonumber \\
%& = & \frac 1{15} g_{22}(t_2)g_{33}(t_3) + \frac 19 g_{22}(t_2)g_{32}(t_3) + \frac 13 g_{22}(t_2)g_{31}(t_3) + \frac 13 g_{21}(t_2)g_{22}(t_3) + g_{21}(t_2)g_{21}(t_3)\nonumber
%\end{eqnarray}
\begin{figure}[ht]
  \centering
  \subfigure[$(0,0,0)$]{
    \label{fig:4a}     %% label for 1st subfigure
    \includegraphics[width=2in]{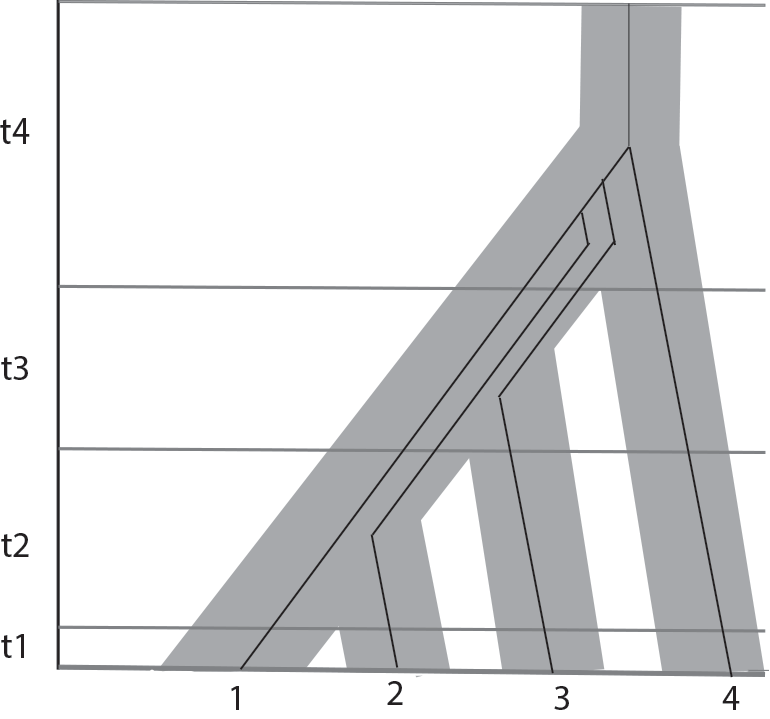}}
  \hspace{0.1in}
  \subfigure[$(0,0,1)$]{
    \label{fig:4b}     %% label for 2nd subfigure
    \includegraphics[width=2in]{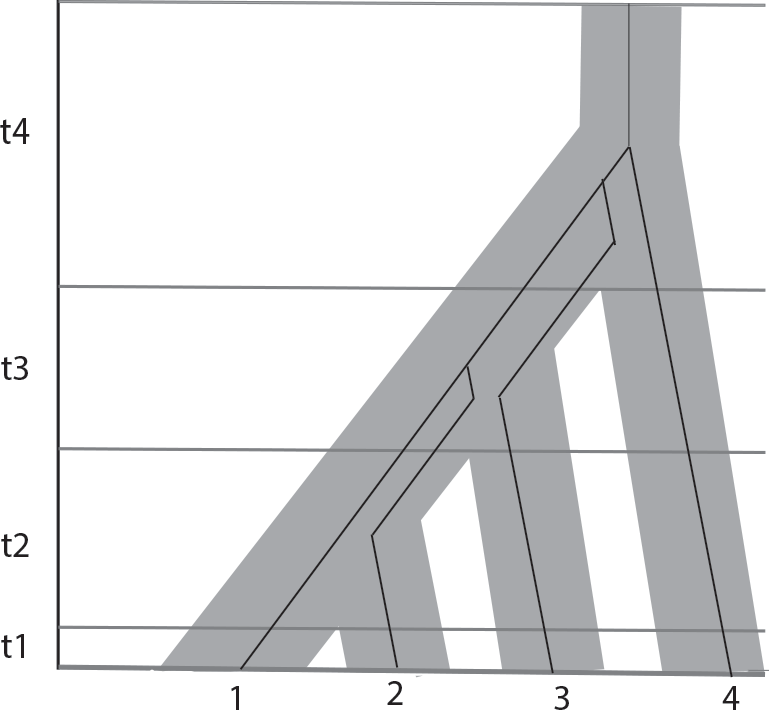}}
      \hspace{0.1in}
  \subfigure[$(0,0,2)$]{
    \label{fig:4c}     %% label for 3rd subfigure
    \includegraphics[width=2in]{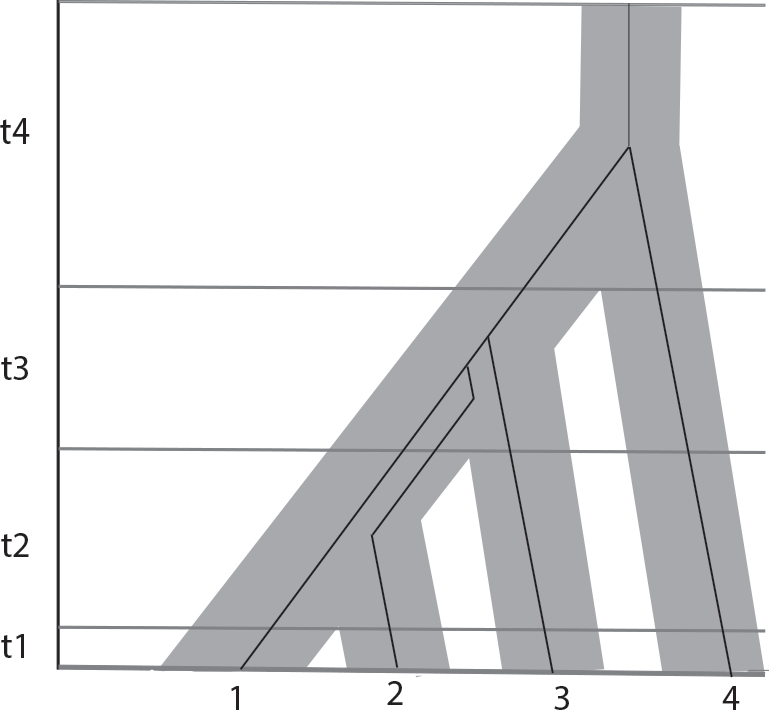}}
  \hspace{0.1in}
  \subfigure[$(0,1,1)$]{
    \label{fig:4d}     %% label for 4th subfigure
    \includegraphics[width=2in]{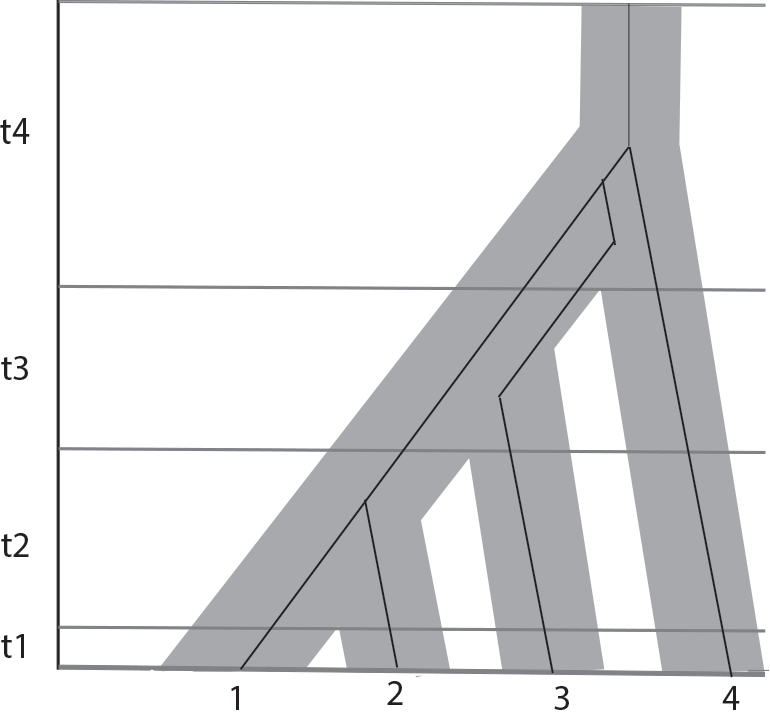}}
  \hspace{0.1in}
  \subfigure[$(0,1,2)$]{
    \label{fig:4e}     %% label for 5th subfigure
    \includegraphics[width=2in]{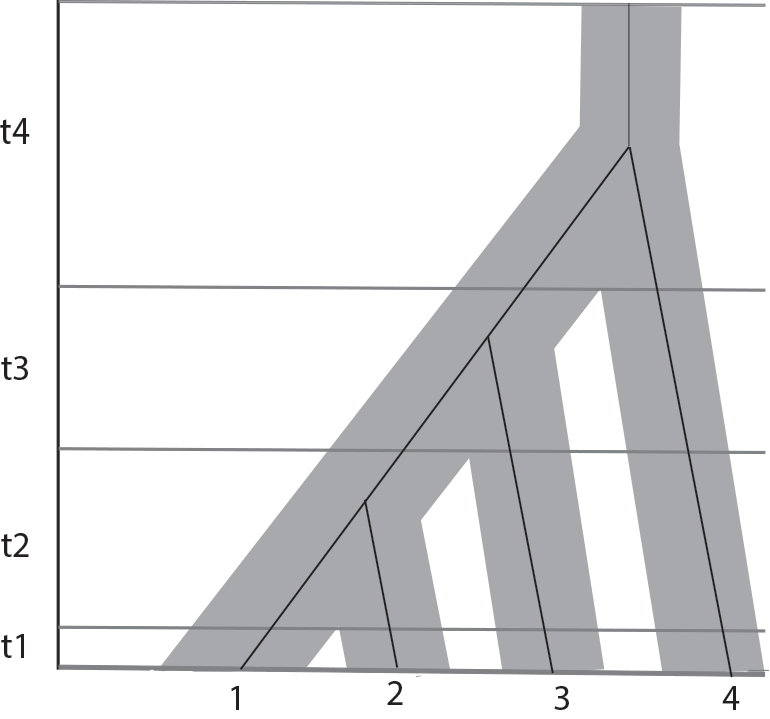}}
  \caption{Gene trees with $d(T_s, T_e)=0$ and their coalescent timelines $(i_1, i_2, i_3)$}
  \label{fig:size4ex}     %% label for entire figure
{\footnotesize }
\end{figure}
\end{example}

By Theorem \ref{thm:probd0}, if we have larger $t_k$ for $k = 1, \cdots,
n$, then we have higher probability that the species tree $T_s$ and
its gene tree $T_g$ generated under the coalescent given $T_s$ have
the same tree topology. In addition,  since $k$-IC is the
$l_{\infty}$ norm of the vector in $\R^{n \choose 2}$, the path
difference is the $l_2$ norm of the vector in $\R^{n \choose 2}$, and
the edge difference is the $l_1$ norm of the vector in $\R^{n \choose
  2}$, $k$-IC distance tree metric can be used 
for the upper
bound for the path difference tree metric and the edge difference tree
metric by Remark \ref{bounds}.  Thus, by Lemmas \ref{coon1}
and \ref{coon2}, if we have larger $t_k$ for $k = 1, \cdots,
n$, then the distributions of tree distance metric $d_e$,
$d_p$ and $d_k$ between $T_s $ and $ T_g$ are skewed from right.

\section{Simulations}\label{simulations}

First we have conducted simulations study on the three tree distances,
the edge difference, path difference, and precise $K$-IC distances
between two unrooted random trees with $12$ leaves.  We have conducted
a simulation study similar to what \cite{Steel1993} did (Figure 6 on
their paper).  We generated $10,000$ unrooted random trees with $12$
leaves using the function 
{\tt rtree} from {\tt R} package {\tt ape} \cite[]{ape}.  Then  for
each distance measure $d_e, \, d_p, \, d_k$  we computed a histogram.
In order to compare a histogram with each other we normalized the
distances so that they scale from $0$ to $10$.  The results are shown
in Figure \ref{fig:sim1}.  We also conducted the same simulations with the
function {\tt rcoal} from {\tt ape} and we have obtained basically the
same results. 

\begin{figure}
        \centering
        \begin{subfigure}[A histogram of $d_e$ between two unrooted random trees
                with 12 leaves. We scale $d_e$ from $0.0$ to $10.0$
                so that we can compare to the other distance measures.]{\label{fig:edge}%{0.3\textwidth}
                \includegraphics[width=2in]{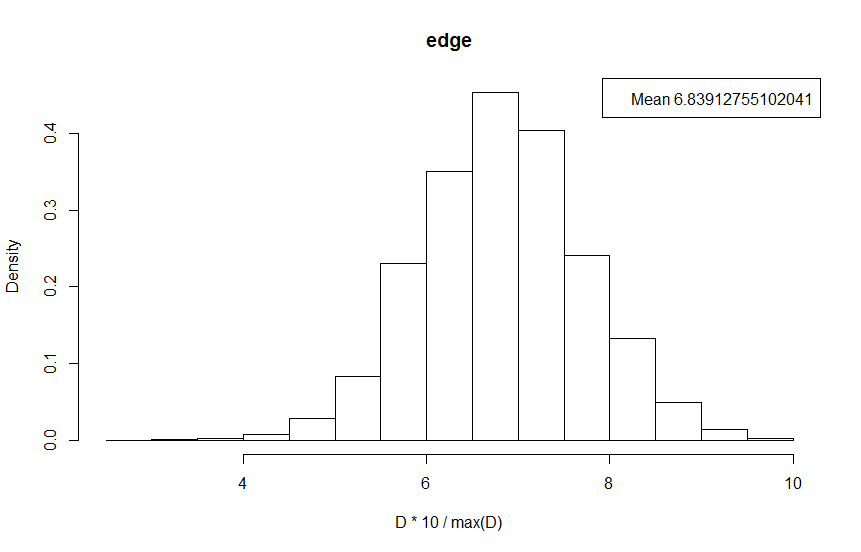}}
                % \caption{A histogram of $d_e$ between two unrooted random trees
                % with 12 leaves. We scalee $d_e$ from $0.0$ to $10.0$
                % so that we can compare to the other distance measures.}
                % \label{fig:edge}
        \end{subfigure}
        ~ %add desired spacing between images, e. g. ~, \quad, \qquad, \hfill etc.
        \begin{subfigure}[A histogram of $d_p$ between two unrooted random trees
                with 12 leaves. We scalee $d_p$ from $0.0$ to $10.0$
                so that we can compare to the other distance measures.]%{0.3\textwidth}
                { \label{fig:path} \includegraphics[width=2in]{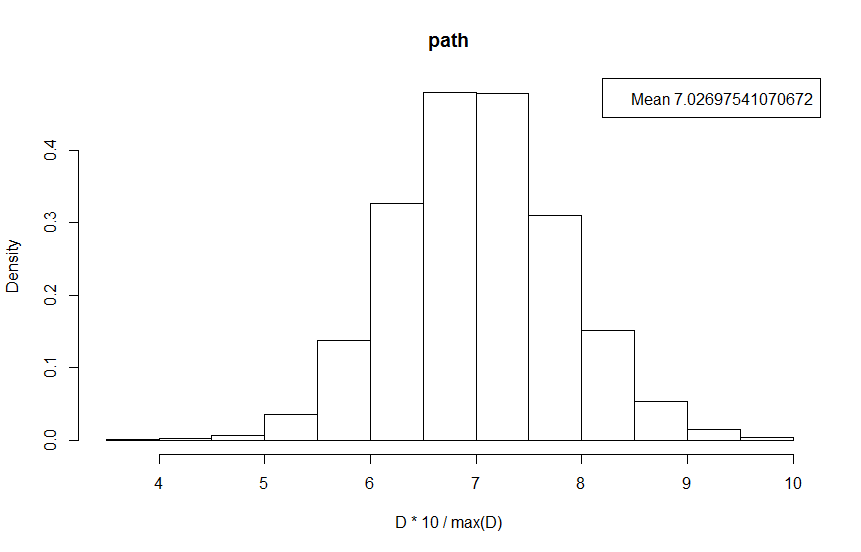}}
                % \caption{A histogram of $d_p$ between two unrooted random trees
                % with 12 leaves. We scalee $d_p$ from $0.0$ to $10.0$
                % so that we can compare to the other distance measures.}
                % \label{fig:path}
        \end{subfigure}
        ~ %add desired spacing between images, e. g. ~, \quad, \qquad, \hfill etc.
        \begin{subfigure}[A histogram of $d_k$ between two unrooted random trees
                with 12 leaves. We scalee $d_k$ from $0.0$ to $10.0$
                so that we can compare to the other distance measures.]%{0.3\textwidth}
                { \label{fig:KIC} \includegraphics[width=2in]{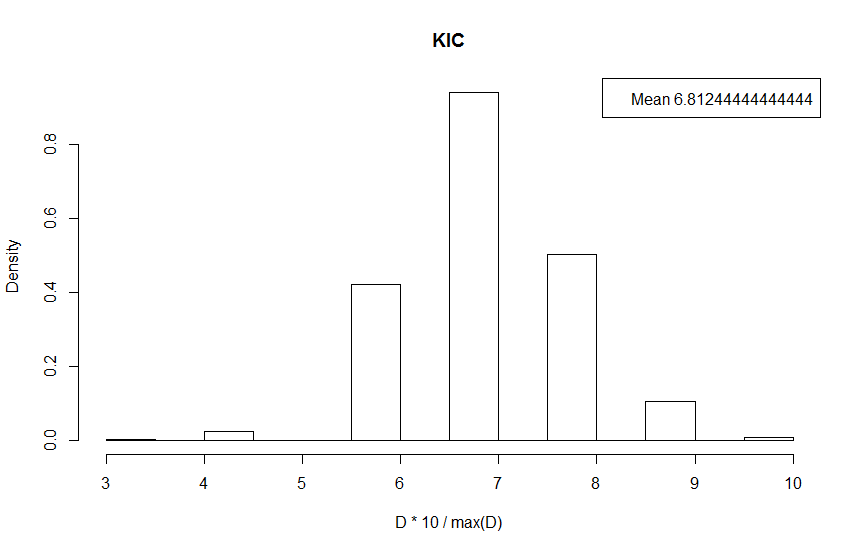}}
                % \caption{A histogram of $d_k$ between two unrooted random trees
                % with 12 leaves. We scalee $d_k$ from $0.0$ to $10.0$
                % so that we can compare to the other distance measures.}
                % \label{fig:KIC}
        \end{subfigure}
        \caption{We generated $10,000$ random trees using the function
        {rtree} from {\tt ape}.  }\label{fig:sim1}
\end{figure}

In the second simulation part, we conducted a simulation study on the
distributions of $d_e, \, d_p, \, d_k$ between the caterpillar species
tree and a random gene tree generated from the coalescent process with
the species tree.  We use the software {\tt Mesquite}
\cite[]{mesquite275} to generate caterpillar species trees with 5
leaves, 6 leaves, 7 leaves and 8 leaves, respectively under the Yule
process. Then we simulate 10,000 gene trees within each species
tree. For all the trees in the simulation, they have the same
parameters, that is the effective population size $N_e=30,000$ and
species depth$=1,000$. For each kind of trees with certain number of
leaves, we then calculated three different kinds of distances between
the gene trees and species trees. Table \ref{sim:prob0} shows the
proportions of 0 and 1 distances in each of the three distances for the rooted trees with 
5 leaves, 6 leaves, 7 leaves and 8 leaves. Figures \ref{sim_coal1},
\ref{sim_coal2}, and \ref{sim_coal3}
show the histograms of three kinds of distances for trees with 5
leaves, 6 leaves, 7 leaves and 8 leaves. 

\begin{table}[!h]
\begin{center}
\resizebox{0.7\textwidth}{!}{%
% \selectlanguage{english}%
\begin{tabular}{c|ccc}
5 leaves & Sample Proportion & Mean Distance & Standard Deviation\tabularnewline
\hline 
$d_k$ = 0 & 0.9543 & \multirow{2}{*}{0.0457} & \multirow{2}{*}{0.2088}\tabularnewline
$d_k$ = 1 & 0.0457 &  & \tabularnewline
\hline 
$d_e$ = 0 & 0.9543 & \multirow{2}{*}{0.2742} & \multirow{2}{*}{1.2531}\tabularnewline
$d_e$  = 1 & 0 &  & \tabularnewline
\hline 
$d_p$  = 0 & 0.9543 & \multirow{2}{*}{0.1119} & \multirow{2}{*}{0.5116}\tabularnewline
$d_p$  = 1 & 0 &  & \tabularnewline
\end{tabular}}%\foreignlanguage{british}{}}\\
% \label{table:name}
\\
\vskip 0.2in

% \selectlanguage{british}%
 \resizebox{0.7\textwidth}{!}{%
% \\
% \selectlanguage{english}%
\begin{tabular}{c|ccc}
6 leaves & Sample Proportion & Mean Distance & Standard Deviation\tabularnewline
\hline 
$d_k$= 0 & 0.9007 & \multirow{2}{*}{0.1025} & \multirow{2}{*}{0.3137}\tabularnewline
$d_k$ = 1 & 0.0961 &  & \tabularnewline
\hline 
$d_e$= 0 & 0.9007 & \multirow{2}{*}{0.8200} & \multirow{2}{*}{2.4899}\tabularnewline
$d_e$ = 1 & 0 &  & \tabularnewline
\hline 
$d_p$ = 0 & 0.9007 & \multirow{2}{*}{0.2869} & \multirow{2}{*}{0.8682}\tabularnewline
$d_p$ = 1 & 0 &  & \tabularnewline
\end{tabular}}%\foreignlanguage{british}{}}
\\
\vskip 0.2in

% \selectlanguage{british}%
 \resizebox{0.7\textwidth}{!}{%
% \selectlanguage{english}%
\begin{tabular}{c|ccc}
7 leaves & Sample Proportion & Mean Distance & Standard Deviation\tabularnewline
\hline 
$d_k$ = 0 & 0.4824 & \multirow{2}{*}{0.6842} & \multirow{2}{*}{0.7420}\tabularnewline
$d_k$  = 1 & 0.3516 &  & \tabularnewline
\hline 
$d_e$  = 0 & 0.4824 & \multirow{2}{*}{6.5920} & \multirow{2}{*}{6.7687}\tabularnewline
$d_e$  = 1 & 0 &  & \tabularnewline
\hline 
$d_p$= 0 & 0.4824 & \multirow{2}{*}{1.9531} & \multirow{2}{*}{1.9685}\tabularnewline
$d_p$ = 1 & 0 &  & \tabularnewline
\end{tabular}}%\foreignlanguage{british}{}}
\\
\vskip 0.2in

% \selectlanguage{british}%
 \resizebox{0.7\textwidth}{!}{%
% \selectlanguage{english}%
\begin{tabular}{c|ccc}
8 leaves & Sample Proportion & Mean Distance & Standard Deviation\tabularnewline
\hline 
$d_k$ = 0 & 0.0760 & \multirow{2}{*}{1.8490} & \multirow{2}{*}{0.9002}\tabularnewline
$d_k$ = 1 & 0.2639 &  & \tabularnewline
\hline 
$d_e$ = 0 & 0.0760 & \multirow{2}{*}{20.2859} & \multirow{2}{*}{9.0730}\tabularnewline
$d_e$ = 1 & 0 &  & \tabularnewline
\hline 
$d_p$= 0 & 0.0760 & \multirow{2}{*}{5.1175} & \multirow{2}{*}{2.0716}\tabularnewline
$d_p$ = 1 & 0 &  & \tabularnewline
\end{tabular}}%\foreignlanguage{british}{}}\selectlanguage{british}%
\caption{The
proportions of 0 and 1 distances in each of the three distances $d_e,
\, d_p, \, d_k$ for
the rooted trees with 
5 leaves, 6 leaves, 7 leaves and 8 leaves.}\label{sim:prob0}
\end{center}
\end{table}

\begin{figure}[!h]
%\begin{minipage}[t]{1\columnwidth}%
\begin{center}
\includegraphics[scale=0.7]{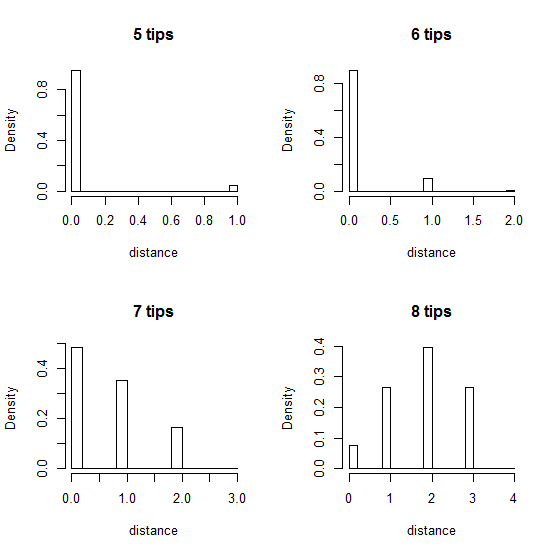}%
%\end{minipage}
\end{center}
\caption{Histogram of $d_k$ for the caterpillar species tree and a
  random tree generated from the coalescent process}\label{sim_coal1}
\end{figure}

\begin{figure}[!h]
%\begin{minipage}[t]{1\columnwidth}%
\begin{center}
\includegraphics[scale=0.7]{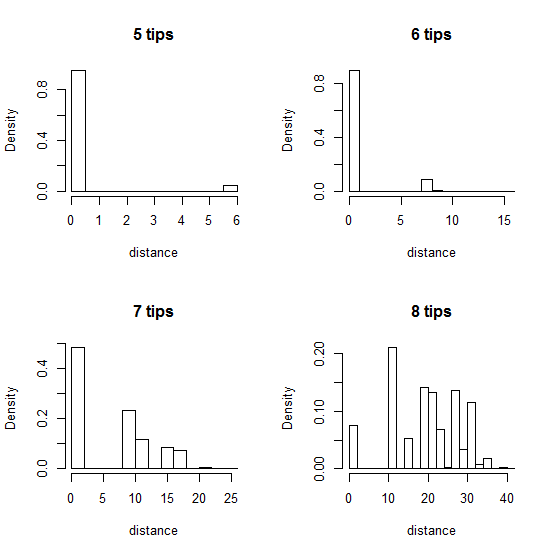}%
%\end{minipage}
\end{center}
\caption{Histogram of $d_e$ for the caterpillar species tree and a
  random tree generated from the coalescent process}\label{sim_coal2}
\end{figure}

\begin{figure}[!h]
%\begin{minipage}[t]{1\columnwidth}%
\begin{center}
\includegraphics[scale=0.7]{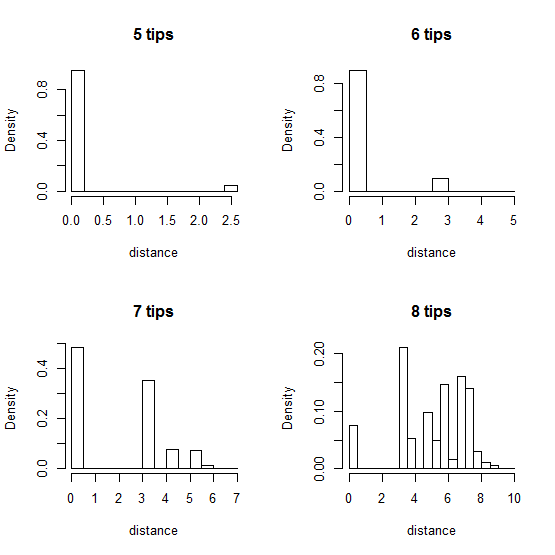}%
%\end{minipage}
\end{center}
\caption{Histogram of $d_p$ for the caterpillar species tree and a
  random tree generated from the coalescent process}\label{sim_coal3}
\end{figure}

\section{Discussion}\label{dis}
While many tree distances measures between trees are hard to compute
(see Remark \ref{rm:np}) tree distances $d_e, \, d_p, \, d_k$ can be
computed in polynomial time in $n$.  Today,  we can generate
huge numbers of DNA sequences from genomes using new generation
sequencing techniques and
they can generate tens of millions base pairs of DNA sequences. In
order to conduct phylogenomics analysis on genome data sets we need
fast tree distances, such as  $d_e, \, d_p, \, d_k$.  However, in
order to understand statistical phylogenomics analysis on genome data
sets with thesis tree distances, we have to understand distribution of
these distances.  

In this paper we have shown some theoretical and simulation results on
the distributions of tree distances $d_e, \, d_p, \, d_k$ between
unrooted random trees with $n$ leaves and between the caterpillar
species tree and a random rooted gene tree with $n$ leaves generated from the
coalescent process with the species tree.  

The distributions of tree distances $d_e, \, d_p, \, d_k$ between
unrooted random trees with $n$ leaves seem to be symmetric and we have
conducted some goodness of fit test with the Gaussian distribution.
However, the null hypothesis (the distribution fits with the Gaussian
distribution) seems to be rejected (with the number of trees equals to
$10,000$), so it would be interesting and useful to know the
asymptotic distributions of  $d_e, \, d_p, \, d_k$ between
unrooted random trees with $n$ leaves.

% In Figure \ref{fig:sim1}, we observe that the distributions of tree distances $d_e, \, d_p, \, d_k$ between
% unrooted random trees with $n$ leaves seem to be symmetric and resemble Gaussian distributions.
% In order to verify this, we used two goodness-of-fit tests to test the normality of the datasets: Shapiro-Wilk test and QQplot.
% However, the null hypothesis (the distribution fits with the Gaussian
% distribution) were rejected (with the number of trees equals to
% $10,000$) by both tests, so it would be interesting and useful to know the
% asymptotic distributions of  $d_e, \, d_p, \, d_k$ between
% unrooted random trees with $n$ leaves.  
% \begin{problem}
% What are the asymptotic distributions of $d_e, \, d_p, \, d_k$ between unrooted random trees with $n$ leaves?
% \end{problem}

In Theorem \ref{thm:probd0}, we have shown explicitly the probability of the tree distance $d_e,
\, d_p, \, d_k$ between caterpillar species
tree with $n$ leaves and a random gene tree with $n$ leaves
distributed with the coalescent process with the species tree equals
to zero. Note here the species tree is assumed to be caterpillar because $d_k$ between two trees can reach its upper bound only if one of them is caterpillar. Figure \ref{sim_coal1}, Figure \ref{sim_coal2} and Figure \ref{sim_coal3} show us that when the sizes of trees get larger, the centers and variation of non-zero distances also become larger, but zero is the only distance value that always guarantee a positive probability for all three types of distances. We are also interested in the computing the probability of $d_k$ being one, which is generally zero for $d_e$ and $d_p$ (see Table \ref{sim:prob0}).
% We have shown explicitly the probability of the tree distances $d_e,
% \, d_p, \, d_k$ between caterpillar species
% tree with $n$ leaves and a random gene tree with $n$ leaves
% distributed with the coalescent process with the species tree equals
% to zero. 
However we do not know many aspects of the tree distance $d$ (one of the
distances $d_e, \, d_p, \, d_k$) between them as $n \to \infty$.
Thus, we have the following questions.
\begin{problem}
Consider the tree distances $d_e,
\, d_p, \, d_k$ between caterpillar species
tree with $n$ leaves and a random gene tree with $n$ leaves
distributed with the coalescent process with the species tree.  What
is the expectation of the tree distance $d$ (one of the
distances $d_e, \, d_p, \, d_k$) between them?  How about variance?
Can we say anything about the expectation asymptotically?
\end{problem}

\section{Acknowledgements}
The authors would like to thank the referees for very useful comments
to improve the manuscript.

\pagebreak
%\clearpage
%\setcounter{page}{1}
%\bibliographystyle{amnat2}
%\bibliographystyle{natbib}
%\bibliographystyle{plain}
\bibliographystyle{decsci}%{plainnat}

\bibliography{gene}
%\bibliographystyle{named}
%\bibliographystyle{plain}
%\bibliographystyle{ama}
%\bibliography{gene}
%\bibliography{samsi}
\end{document}